\documentclass{article}

\usepackage{balance} 
\usepackage{subcaption}
\usepackage{bbm} 
\usepackage{diagbox}

\usepackage{graphicx}
\usepackage[table]{xcolor}
\usepackage[utf8]{inputenc}
\usepackage{amsmath}
\usepackage{amssymb} 
\usepackage{amsthm}
\newtheorem{theorem}{Theorem}[section]
\newtheorem{proposition}[theorem]{Proposition}
\newtheorem{lemma}[theorem]{Lemma}
\newtheorem{corollary}[theorem]{Corollary}

\newtheorem{definition}[]{Definition}
\usepackage{hyperref} 
\usepackage{enumerate}

\title{Majority Illusions in Large Random Networks}

\author{Maaike Venema-Los\\
\texttt{m.d.los@rug.nl}
\and Zo\'{e} Christoff\\
\texttt{z.l.christoff@rug.nl}
\and
Serte Donderwinkel\\
\texttt{s.a.donderwinkel@rug.nl}
\and
Davide Grossi\\
\texttt{d.grossi@rug.nl}}
\date{}

\newcommand{\ser}[1]{#1}

\newtheorem{remark}[]{Remark}

\begin{document}




\maketitle 


\textit{An agent in a social network is under majority illusion when the majority among its neighbors differs from the majority in the whole network.
We study analytically the probability of the occurrence of majority illusions on large random networks. We derive 0-1 laws for graph sequences that locally converge, using different types of random opinion-assignment functions. Our results show that in large random networks with uniform random opinion-assignment the probability of majority illusions goes to $0$. However, when opinions depend on nodes' degrees the probability may go to $1$. 
From a methodological point of view, our work illustrates how local convergence techniques from probability theory can be effectively applied to the understanding of opinion-related phenomena---such as majority illusions---in large networks.}

\section{Introduction}

\begin{table}[t]
    \centering    
    \caption{Limits for sequences of Erd\H{o}s--R\'{e}nyi graphs,  continuum scale-free percolation (CSFP) graphs, and the more general class of graphs where the proportion of vertices that are at distance at most 2 from each other goes to 0 ($\boldsymbol{G_n^{>2}}$), with different random colorings (see Section \ref{sec:local_convergence} for explanations on the graphs and colorings). `0' indicates that the probability of the $n$-the graph being in \texttt{Mm} illusion goes to 0 as $n\rightarrow\infty$, `1' that it goes to 1, 0 / 1 that both are possible. Arrows indicate implication direction between results. References to the corresponding results are included.}
    \label{tab:results}
    \begin{tabular}{|c||c|c|c|c|}
    \cline{1-3} \cline{5-5}
         \diagbox[width = 7em]{\textit{Coloring}}{Graphs} & Erd\H{o}s--R\'{e}nyi & CSFP& & $\boldsymbol{G}_n^{>2}$\\ \cline{1-3} \cline{5-5}\noalign{\vskip\doublerulesep
         \vskip-\arrayrulewidth}\cline{1-3} \cline{5-5}
      \textit{i.i.d. random}  &\textbf{0} (Cor.  \ref{cor:iid_coloring})& \textbf{0} (Cor. \ref{cor:iid_coloring})  & $\Leftarrow$ & \textbf{0} (Th. \ref{thm:iid_coloring})\\ \cline{1-3} \cline{5-5}
      \textit{degree-based} &\textbf{0} (Th. \ref{thm:conv_ER_coloring_degree}) & \textbf{1} (Th. \ref{thm:conv_CSFP}) & $\Rightarrow$ & \textbf{0} / \textbf{1}\\ \cline{1-3} \cline{5-5}
    \end{tabular}

\end{table}

When people form opinions or make decisions, they are often influenced by others around them. At the same time, one's direct 
circles 
are
not always a representative sample of the broader community. 
When an agent observes a majority 
opinion (or `color') 
among its direct neighbors which 
differs
from the global majority, 
this agent 
is under so-called \emph{majority illusion} \cite{Lerman2016}. This illusion has been studied from 
different perspectives, among which: the structure of the graphs allowing it 
\cite{Venema-Los2025}; its computational complexity \cite{Grandi2025}; and the complexity of eliminating it 
\cite{fioravantes2025,dippel2025}\footnote{Although this gives a complete view of the study of majority illusions in AI and computer science, the majority illusion is often mentioned in the social science literature, see for example \cite{Lee2019} and \cite{Zhang2019}.}.
However, no work so far has specifically focused on majority illusions on large networks. Even though the study of majority illusions is often motivated in the context of social networks, 
it is hard to understand the behavior of such networks when they grow large,
due to limited data access or computational resources.  We tackle the question: {\em ``What is the probability that majority illusions occur as networks grow large?''}. 
To study this question, we use a notion from probability theory called \emph{local convergence} \cite{Hofstad_2016}. With this tool, we can study properties of large random networks by considering sequences of networks of increasing sizes, and studying convergence behavior of such sequences. For specific types of random networks, there exist results on the limit graphs to which they converge, and especially on what the neighborhood of a random node in the limit graph looks like.
In \cite{Hazra2025} these methods are used to study structural network phenomena on large random graphs. We extend those methods to the study of {\em colored} random graphs, where colors represent opinions. Note that in the context of social networks, focusing on such graphs that locally converge is natural: when the total network grows in size, typically the local neighborhood of a node (for example, the average amount of friends someone has) does not change.

\paragraph{Contribution}
Our core contribution consists in deriving asymptotic bounds on the probability of the occurrence of majority illusions in 
locally converging sequences of random graphs, in particular Erd\H{o}s--R\'{e}nyi and continuum scale-free percolation graphs.
More generally, our results illustrate how local convergence techniques from probability and network theory can be fruitfully used to analyze 
network phenomena involving opinion formation, which are of special relevance for multi-agent systems. 

Table \ref{tab:results} summarizes the specific results we obtain. We show that on random graph sequences that locally converge, under uniform random colorings the probability of a majority of nodes being under majority illusion tends to 0. However, with 
a coloring that is based on the degree of the nodes, this probability depends on the graph structure. In Erd\H{o}s--R\'{e}nyi graphs (`as random as possible') it still goes to 0 as the graph size goes to infinity, but in the continuum scale-free percolation model, which exhibits properties that are often observed in real networks, we can find parameters such that the probability of a majority of nodes being under majority illusion goes to 1.



\paragraph{Related work and modeling choices.} We briefly motivate the choice of random models 
used in this paper, 
based on the relevant related literature. In a nutshell, when studying majority illusions, two types of correlations are relevant: 
the correlation between a node's degree and how often it is observed,  
and the correlation between a node's degree and its color.

Indeed, first of all, the concept of majority illusion was 
initially introduced in \cite{Lerman2016} as a variation of the \emph{friendship paradox} \cite{Feld1991}. The friendship paradox refers to the observation that on average, a node's neighbors have more neighbors than the node itself. 
This counter-intuitive aspect of networks can be explained by the fact that nodes with high degrees are observed more often as neighbors than nodes with low degrees, thereby skewing the average degree of a neighbor.
The \emph{friendship bias} is a measure of this disagreement: a node's bias is the average degree of 
its 
neighbors divided by 
its own degree. 
The friendship bias is generally higher in networks with a heterogeneous degree distribution, 
a feature often observed in real (social) networks \cite{Easley2010_NCM}, than in networks where degrees are more similar. In regular networks, where all degrees are equal, the friendship bias is 0.

Second, according to \cite{Lerman2016}, the friendship paradox is an underlying phenomenon explaining the majority illusion
, together with a correlation between nodes' degree and color. Intuitively, because nodes with larger degrees are observed more frequently,
when 
a node's color
depends on its degree,  this causes a disproportional observation of the nodes of one color, resulting in majority illusions. 
The authors of 
\cite{Lerman2016} work this out with computational experiments and statistical computations and show that in real and artificial networks, there are three factors that jointly influence the amount of nodes under majority illusion: degree assortativity and degree distribution heterogeneity  (which together determine the friendship bias),  and correlation between degree and color. 
They find that in the networks they study, nodes are more often under majority illusion when the network exhibits a \emph{disassortative} structure regarding degrees (nodes are likely to connect to nodes with dissimilar degrees), the  degree distribution is \emph{heterogeneous} (for example as in a scale-free network), and when there is \emph{more correlation} between the degree of a node and its color. The intuition behind this is that in such networks, there are few nodes with high degree, which are mostly of the minority color, and mostly connected to many nodes with low degree. At the same time, there are many nodes with low degree, which will observe mostly the high-degree nodes, and since those are the minority color the low-degree nodes will be under majority illusion\footnote{Note, however, that the combination of degree disassortativity, degree-color correlation and heterogeneous degree distribution is not a necessary condition for majority illusions to occur. As shown in \cite{Venema-Los2025}, is is possible for a majority of nodes to be under majority illusion on regular networks too. 
Furthermore, it is possible that a majority of nodes are under a weak form of majority illusion (where seeing a tie when there is a strict majority is also considered an illusion), on \emph{any} network.}. 

Furthermore, experiments on simulated network data in \cite{Venema-Los2025} show that if there is no correlation between the degree of a node and its color (they use a uniformly random coloring), the effect of the degree assortativity on the amount of nodes under majority illusion is small.
Hence, the correlation between color and degree seems to have an effect on the likelihood of majority illusions. Besides that, 
there are more reasons to study colorings that correlate with degree. 
There are different examples of attributes in real networks that are correlated with the degree of nodes. It is a simple proxy for, for example, highly connected individuals voting more liberally \cite{Lazer2010}. 
Furthermore, there is evidence that several properties of individuals in networks correlated with the size of their neighborhood, for example:
happiness or well-being\cite{Huang2021,Kang2023}; likeliness to vote \cite{Kernell2023}.
 
The above-mentioned relationship between friendship paradox and degree-coloring correlation on the one hand and majority illusions on the other hand, together with the real world examples of attributes correlating with degree, are precisely what motivate our modeling choices for random graphs and random colorings in this paper.
We consider two 
types of random graphs
and two 
types of colorings. 
For the graphs,  \emph{Erd\H{o}s--R\'{e}nyi graphs} are used as a baseline for random graphs, whereas 
\emph{continuum scale-free percolation} 
graphs 
conveniently 
come with 
heterogeneous (indeed, `scale-free') degree distribution and an underlying geometric structure that enhances degree disassortativity.
On 
both types of graphs, 
we use a basic \emph{uniform random coloring}, as well as a \emph{coloring that correlates with the degree of nodes}. Both types of graphs and types of colorings will be introduced in more detail in Section \ref{sec:preliminaries}.

\section{Preliminaries}\label{sec:preliminaries}
In this section we first 
introduce 
key definitions about majority illusions. We then present---in the necessary level of detail---relevant definitions and results about large random graphs, which we are going to leverage in our results.

\subsection{Graphs and Majority Illusions}
Let $G = \langle V, E\rangle$ 
be a graph/network (we will use those words interchangeably) where $V$ is a set of nodes and $E$ a set of edges, and let $G^\gamma = \langle V, E, \gamma\rangle$ be a two-colored graph where $\gamma:V\rightarrow\{blue, red\}$ is a color mapping.
Throughout the paper, when the (colored) graph is clear from the context, we will use $V$ and $E$ for the underlying sets of vertices and edges without further introduction.

We write $\boldsymbol{G} = \langle V, E \rangle$ for 
a random graph, and  $\boldsymbol{G}^\gamma = \langle V, E, \gamma \rangle$  for a colored random graph, where $E$ and $\gamma$ are random variables. We denote by  $(\boldsymbol{G}_n)$ a sequence of random graphs, and by 
$(\boldsymbol{G}_n^\gamma)$ a sequence of colored random graphs,  
such that the size $n=|V|$ increases every step in the sequence. We write $\gamma_n$ to refer to the coloring of the 
graph of size $n$.

The majority winner $M^S$ of a set 
$S\subseteq V$
in colored graph 
$G^\gamma 
$
 is the color ($\in \{red, blue\}$) that occurs most often in $S$, or, if both colors occur equally often, $M^S=tie$. We write $M^S_t$ for the majority winner of $S$ at time $t$. 
Following \cite{Venema-Los2025}, 
we say that 
an agent is under \emph{majority illusion} if its neighborhood majority winner and the global majority winner are two different colors.
We write $N_i$ for the set of $i$'s neighbors
.

\begin{definition}[Majority illusion \cite{Venema-Los2025,Grandi2025}
]
\label{def:maj-ill-strict}
Given 
a 
colored 
graph $G^\gamma
$, an agent $i\in V$ is under \emph{majority illusion}   (\texttt{m} illusion)
if $M^{N_i}\not = \text{tie}$  and $M^{V}\not = \text{tie}$ and $M^{N_i} \not = M^{V}$.  A graph is in a  \emph{Majority-majority illusion}  (\texttt{Mm} illusion) if more than half of the agents are under majority illusion. 
\end{definition}
See Figure \ref{fig:Mm_illusion} for an illustration of the definition of majority illusion.
\begin{figure}[t]
    \centering
    \includegraphics[width=0.4\linewidth]{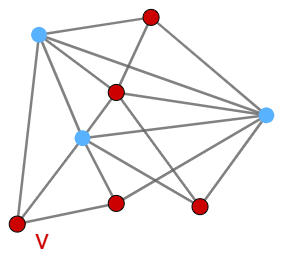}
    \caption{(From \cite{Venema-Los2025}) A network where node 
    $v$ is  under \emph{majority illusion} (it has more blue neighbors, while there are more red nodes in the entire network). Moreover, since all red nodes (a majority) are under majority illusion, the network is in \texttt{Mm} illusion.}
    \label{fig:Mm_illusion}
\end{figure}

\subsection{Large Random Graphs}\label{sec:prelim_large_graphs}

\subsubsection{Key definitions}
An essential work that discusses local convergence on random graphs is that of Van der Hofstad \cite{Hofstad_2016}, Volume 2, which we use as reference for  definitions. 
The idea of \emph{local convergence on random graphs} is that sequences of random graphs (in our case the size of which increases every step in the sequence) 
converge to a rooted limit graph. This means that if one takes a random point on a graph far in the sequence, \ser{the distribution of the neighborhood of that point gets very close to the distribution of the } neighborhood of a marked point, the root, in the limit graph. 
To define local convergence in probability of random graphs precisely, we need some auxiliary definitions.

Let $d(x,y)$ be the distance between two vertices $x$ and $y$ in a graph, i.e. the number of edges in the shortest path between $x$ and $y$. 
Sometimes, we need to refer to a specific node in a graph:
A rooted graph $(G,v)$ is a graph $G = \langle V, E\rangle$ with a designated vertex $v\in V$.  
We want to define local convergence of graph sequences: the further in the sequence, the more similar two subsequent graphs in the sequence are. Therefore, we need to be able to measure similarity, or distance, of (colored) rooted graphs.
To define the metric, we need the notions of graph isomorphism, rooted $k$-balls and local finiteness.
Two graphs are isomorphic if they are equal up to the placement of the nodes:
\begin{definition}[Graph isomorphism \cite{Hofstad_2016}]
    Two graphs $G_1=\langle V_1, E_1\rangle$, $G_2=\langle V_2, E_2\rangle$ are \emph{isomorphic} (we write $G_1\simeq G_2$) if there exists a bijection $\phi:V_1\rightarrow V_2$ such that $\langle u_1,v_1\rangle \in E_1$ iff $\langle\phi(u_1),\phi(v_1)\rangle\in E_2$. 
     Two rooted graphs  $(G_1, o_1)$, $(G_2, o_2)$ are isomorphic ($(G_1, o_1)\simeq (G_2, o_2)$) if there exists a bijection $\phi:V_1\rightarrow V_2$ such that $\phi(o_1)=o_2$ and $\langle u_1,v_1\rangle \in E_1$ iff $\langle\phi(u_1),\phi(v_1)\rangle\in E_2$. 
    Two colored graphs $G_1^\gamma$ and $G_2^{\gamma}$ are isomorphic if the uncolored graphs  $G_1$ and $G_2$ are isomorphic and there exists a corresponding bijection that preserves the colors of the nodes.
\end{definition}
A $k$-ball is a subgraph of which all nodes are at most a $k$ distance away from the root:
\begin{definition}[Rooted $k$-ball]
    A \emph{$k$-ball} $B^{(G)}_k(o)$ in a graph $G=\langle V, E\rangle$ 
    centered at some $o \in V$, is the set of vertices $v \in V$ such that $d(o, v) \leq k$ and the set of edges between such vertices.
\end{definition}

See Figure \ref{fig:example2} for an example of a graph with a rooted $2$-ball.
    \begin{figure}[t]
        \centering
        \begin{subfigure}[t]{0.45\linewidth}
            \includegraphics[width=\linewidth]{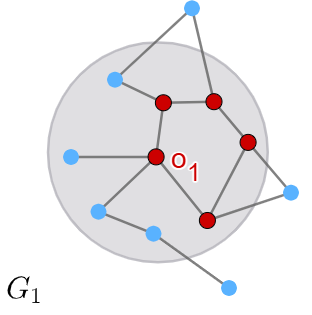}
        \end{subfigure}%
        ~
        \begin{subfigure}[t]{0.45\linewidth}
            \includegraphics[width = \linewidth]{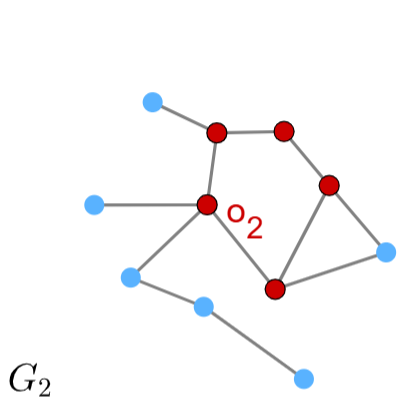}
        \end{subfigure}
        \caption{Illustration of concepts from \ref{sec:prelim_large_graphs}. The gray circle in $G_1$ shows the rooted 2-ball centered at $o_1$. The Benjamini-Schramm distance between $(G_1, o_1)$ and $(G_2, o_2)$ is $\frac{1}{3}$.}
        \label{fig:example2}
    \end{figure}

Even though the graphs that we consider are finite, they may converge to an infinite limit. To ensure that we are still able to define our metric, we do not consider all infinite graphs, but only locally finite ones.

\begin{definition}[Locally finite graphs]
A graph $G$ is \emph{locally finite} if for each $o\in V$ and each integer $k$, it holds that $B_k^{(G)}(o)$ is finite.
\end{definition}

This definition is in particular satisfied by any infinite graph with bounded degrees, such as the lattice in any dimension or the infinite $d$-regular tree. 

We call $\mathcal{G}$ the set of all locally finite rooted graphs and $\mathcal{G}^\gamma$ 
the set of all locally finite colored rooted graphs.
Now we can define a metric for 
locally finite 
rooted graphs:
\begin{definition}[Benjamini--Schramm distance \cite{Benjamini_Schramm_2001}]
For two rooted locally finite connected graphs $(G_1, o_1)$ and $(G_2, o_2)$, the Benjamini--Schramm distance is
    $$d_{BS}((G_1, o_1),(G_2, o_2)) = \frac{1}{k+1},$$ where $k$ is the maximum integer such that $B^{(G_1)}_k(o_1)\simeq B^{(G_2)}_k(o_2)$ (the rooted $k$-ball of $G_1$ is isomorphic to the rooted $k$-ball of $G_2$). If there is no maximum such $k$, i.e. all $k$-balls are isomorphic, then $d_{BS}((G_1, o_1),(G_2, o_2)) =0$.
\end{definition}
For an illustration, consider the two graphs in Figure \ref{fig:example2}. Since the the rooted $2$-ball of $G_1$ is isomorphic to the rooted $2$-ball of $G_2$, but the rooted $3$-balls are not isomorphic, the Benjamini--Schramm distance between the two rooted graphs is $\frac{1}{3}$.

With the Benjamini--Schramm distance, the metric spaces $\mathcal{G}_{BS} =(\mathcal{G},d_{BS})$ on rooted graphs and  $\mathcal{G}^\gamma_{BS} =(\mathcal{G}^\gamma,d_{BS})$ on colored rooted graphs are defined, which naturally gives rise to a notion of convergence for rooted locally finite graphs.

\subsubsection{Local convergence}
Our graphs $G_n$ will not have a given root, but we will instead consider a \emph{random} root, so that local convergence captures a probabilistic description of the neighborhoods in $G_n$, i.e. which \emph{proportion} of neighborhoods has a certain structure. To formally define the local topology for graphs rooted at a random vertex, we need the notion of convergence in probability of random variables. Intuitively, a sequence of random variables $X_n$ converges in probability to $X$ if the probability of $X_n$ deviating from $X$ by a small but fixed amount becomes smaller and smaller.

\begin{definition}[Convergence in probability of random variables]
    Let  $X_1, X_2, ...$ and $X$ be random variables that take values in a metric space $(E,d)$ 
    and are defined on the same probability space. $(X_n)= X_1, X_2, ...$ \emph{converges in probability} towards  $X$ if for all $\epsilon>0$: $\lim_{n\rightarrow\infty}\mathbb{P}(d(X_n,X)>\epsilon)=0.$ We write $(X_n)\xrightarrow{p} X$.
\end{definition}
Now we can define local convergence in probability for unrooted finite graphs in the Benjamini--Schramm space. The idea is that if $n\rightarrow\infty$, the neighborhood around a random point in the $n^{th}$ graph $G_n$ looks more and more like the neighborhood around the root $o$ in the limit graph $G$.
\begin{definition}[Local convergence in probability of 
random 
graphs in $\mathcal{G}_{BS}$]
    A sequence of 
    random finite unrooted graphs $(\boldsymbol{G}_n)$ \emph{converges locally in probability} to a 
    random locally finite rooted graph $(\boldsymbol{G},o)$ (we write $(\boldsymbol{G}_n)\xrightarrow[locally]{p}(\boldsymbol{G}, o)$),  if, for all 
    continuous functions $f$ from $\mathcal{G_{BS}}$ to $\mathbb{R}$,
    $\mathbb{E}[f(\boldsymbol{G}_n, U_n)|\boldsymbol{G}_n]\xrightarrow{p} f(\boldsymbol{G},o),$
    where $U_n$ is a uniformly random 
    vertex in $\boldsymbol{G}_n$.
\end{definition}

Similarly, we can define local convergence in probability of colored graphs:
\begin{definition}[Local convergence in probability of random graphs in $\mathcal{G}^\gamma_{BS}$]
    A sequence of random finite colored unrooted graphs $(\boldsymbol{G}_n^{\gamma})$ 
    converges locally in probability to random locally finite colored rooted graph $(\boldsymbol{G}^{\gamma},o)$ (we write $(\boldsymbol{G}^{\gamma}_n)\xrightarrow[locally]{p}(\boldsymbol{G}^{\gamma}, o)$),  if for all continuous functions $f$ from $\mathcal{G_{BS}}$ to $\mathbb{R}$,
    $\mathbb{E}[f(\boldsymbol{G}^{\gamma}_n, U_n)|\boldsymbol{G}^{\gamma}_n]\xrightarrow{p} f(\boldsymbol{G}^{\gamma},o),$
    where $U_n$ is a uniformly random point in $\boldsymbol{G}^{\gamma}_n$.
\end{definition}

In our case, the relevant continuous functions study the color of a node and the majority among its neighbors. To be precise, we consider   $f=\mathbbm{1}_{\{o\text{ is red}\}}$ and $g=\mathbbm{1}_{\{o\text{ has a majority of red neighbors}\}}$. These functions are continuous, because we can choose small enough $\delta_f$ and $\delta_g$ (namely $\delta_f = 1$ and $\delta_g = \frac{1}{2}$) such that if the distance between two rooted graphs $(G_1^{\gamma}, o_1)$ and $(G_2^{\gamma}, o_2)$ is at most $\delta_f$, then $f(G_1^{\gamma}, o_1) = f(G_2^{\gamma}, o_2)$ and if the distance between $(G_1^{\gamma}, o_1)$ and $(G_2^{\gamma}, o_2)$ is at most $\delta_g$, then $g(G_1^{\gamma}, o_1) = g(G_2^{\gamma}, o_2)$.
Then, 
   $ \mathbb{E}[f(\boldsymbol{G}_n^{\gamma}, U_n)|\boldsymbol{G}_n^{\gamma}] =\mathbb{P}(U_n\text{ is red}|\boldsymbol{G}_n^{\gamma}) 
    =\frac{\#\{v\in \boldsymbol{G}_n^{\gamma}: v\text{ is red} \}}{ \#\{v\in \boldsymbol{G}_n^{\gamma} \}}$.
Thus, if $(\boldsymbol{G}_n^{\gamma})$ converges locally in probability to $(\boldsymbol{G}^\gamma,o)$ then the proportion of red vertices in $(\boldsymbol{G}_n^{\gamma})$ converges in probability to the probability that $o$ is red in $\boldsymbol{G}^\gamma$. Similarly, by continuity of $g$, if $(\boldsymbol{G}_n^{\gamma})$ converges locally in probability to $(\boldsymbol{G}^\gamma,o)$ then the proportion of vertices in $(\boldsymbol{G}_n^{\gamma})$ with a majority of red neighbors converges in probability to the probability that $o$ has a majority of red neighbors in $\boldsymbol{G}^\gamma$. 


After having introduced the concept of local convergence, we will now explain in detail the network models and colorings that we use.

\paragraph{Random network models}
We consider two random network models: the Erd\H{o}s--R\'{e}nyi model and the continuum scale-free percolation model. We use the Erd\H{o}s--R\'{e}nyi model \cite{ErdosRenyi} 
as a baseline random network model. It has only two parameters: the number of nodes $n$ and the probability $p$ 
that any pair of nodes are connected.

A central result on the convergence of  
random graphs is proven in \cite[Theorem 2.18]{Hofstad_2016}: 
\begin{theorem}[Convergence of Erd\H{o}s--R\'{e}nyi graphs \cite{Hofstad_2016}]\label{thm:conv_ER}
Let $\boldsymbol{G}_{n,p}$ be the Erd\H{o}s--R\'{e}nyi graph, which has vertex set $\{1,\dots,n\}$ and each edge is present independently with probability $p$. For $c\in \mathbb{R}$, let $T_{\operatorname{Poisson}(c)}$ be a Bienaymé--Galton--Watson tree with a $\operatorname{Poisson}(c)$ offspring distribution, which is a random tree in which, for each vertex, its number of children are an independent sample with law $\operatorname{Poisson}(c)$. 
Then, for any $c$, 
the sequence $(\boldsymbol{G}_{n, \frac{c}{n}})$ converges locally in probability to $T_{\operatorname{Poisson}(c)}$
: 
    $$(\boldsymbol{G}_{n, \frac{c}{n}})\xrightarrow[locally]{p}(T_{\operatorname{Poisson}(c)},o).$$
\end{theorem}
This means that if  one takes $n$ large enough, then the proportion of vertices in $\boldsymbol{G}_{n,\frac{c}{n}}$ whose local neighborhood takes a particular shape approximates the probability that the first generations of a Bienaymé--Galton--Watson tree with a $\operatorname{Poisson}(c)$ offspring distribution take that shape. (In particular,  since the limit is a tree, for any $k$, the proportion of vertices in $\boldsymbol{G}_{n,\frac{c}{n}}$ that are in a cycle of length at most $k$ goes to $0$.) To understand heuristically why this limit occurs, note that the number of neighbors of any point is distributed as a $\operatorname{Bin}(n-1,c/n)$ random variable, because it is connected to each of the other vertices with probability $c/n$ independently. This is approximately a $\operatorname{Poisson}(c)$ random variable. Then, to find the neighbors of the neighbors, we may again sample independent edges to the remaining vertices. Since we have only seen very few vertices on the scale $n$, this is still approximately a $\operatorname{Poisson}(c)$ random variable for each neighbor. 



In
a recent paper about the friendship paradox \cite{Hazra2025}, 
the authors use Theorem \ref{thm:conv_ER} and similar convergence results to show that in sequences of certain types of random graphs, the friendship bias converges to a limiting value.  
This use of Theorem \ref{thm:conv_ER} inspires our methods. However, while the friendship bias is defined only in terms of the graph structure, the majority illusion also crucially involves the coloring of the nodes. This makes the case of majority illusions more intricate, since we need to show convergence not just for the graph sequence but for the \emph{colored} graph sequence.

Since we study majority illusions, which are mainly interesting because of their implications on social networks, we also consider a model that exhibits properties that are often observed in real social networks. The continuum scale-free percolation (CSFP) model
\ser{ is a random graph model that was introduced in \cite{DeijfenVanderHofstad, CSFP} to capture two important properties of real-world networks: inhomogeneity between nodes (encoding e.g.\ popularity) and latent geometry (encoding e.g.\ physical closeness or, more abstractly, similarity) that both affect connection probabilities.} \ser{Underlying the} CSFP model is a rate $1$ Poisson point process on \ser{the $d$-dimensional box $[0,n]^d$,} where points have random weights. Edges form between nodes with a probability based on distance between nodes and their weights: the probability of an edge between node $x_1$ and $x_2$ with weights $w_1$ and $w_2$ is $1-\exp(\ser{-}\frac{\lambda w_1 w_2}{||x_1 - x_2||^\alpha})$ for $\alpha,\lambda>0$. \ser{(Informally, $\lambda$ controls the overall edge density and $\alpha$ determines how unlikely it is for faraway vertices to connect.)}
The weights of nodes are determined by a Pareto distribution with parameter $\beta$, such that the probability that the weight of a node $i$ is larger than a constant $w>1$ is $\mathbb{P}(W>w)= w^{-\beta}$, \ser{making the weights heavy-tailed.}

In Corollary 4 of \cite{van_der_Hofstad_2023}, the authors show that a sequence of CSFP graphs converges locally in probability to an infinite rooted CSFP graph when $n\rightarrow\infty$.

\paragraph{Random colorings} 
To be able to use the local convergence proof technique, we need to use a coloring that is based only on local properties of nodes. Indeed, if the coloring of a node would depend on non-local properties of the graph, the sequence of colored graphs would not locally converge anymore.
We here consider two types of random colorings of random graphs.
The first is  a basic random color model that does not use any information from the network: the independent identical coloring $\gamma_p$, where each node is colored red with probability $p$ and blue with probability $1-p$. 

The second type of random coloring we study is still independently distributed but based on the network structure: the probability of being red or blue depends on the degree of the node.
For example, in the graphs in Figure \ref{fig:example2}, all nodes with degree at least 3 are red, while all nodes with degree smaller than 3 are blue.
Since the coloring can be random and is based on the degree of a node, we need a distribution for all possible degrees: $\pi=(\pi_d)_{d\ge 0}$ is a sequence of distributions indexed by possible values of the degree $d$. Then, we denote colorings based on degree as $\gamma_{d,\pi}$.
As mentioned earlier, with such degree-based colorings, the hubs that often occur in real networks make majority illusions more likely.

\section{Majority Illusions in Large Random Graphs }
\label{sec:local_convergence}
In this section, we will derive the probability that Majority-majority illusions 
occur in certain types of large random graphs, by studying the local convergence properties of such graphs. 
Because we know that a sequence converges to a specific limit graph and we can prove that on the limit graph majority illusions will occur or not, we can conclude that if one looks far enough in the sequence of random graphs, majority illusions are very (un)likely. 

\paragraph{Contribution}
The remaining of this section is structured as follows. First, in \ref{sec:proof_strategy} we explain the structure of the proofs we use. Then, in \ref{sec:independent_coloring}, we show that with independent colorings, for any graph sequence that locally converges, the probability of \texttt{Mm} illusions goes to 0 as $n\rightarrow\infty$. In \ref{sec:degree-based_coloring} we prove first that, with a coloring based on the degree, in sequences of Erd\H{o}s--R\'{e}nyi graphs the probability of \texttt{Mm} illusions converges to 0, and then that in sequences of CSFP graphs the probability of \texttt{Mm} illusions converges to 1.
We provide the most important proofs in the paper, remaining proofs can be found in Appendix \ref{app:aux/proofs}.

\subsection{Proof Strategy}\label{sec:proof_strategy}
Here we give a high-level explanation of the structure of the proofs that will be instantiated in the next section. 

The idea is that when a sequence of graphs locally converges, we can calculate the \ser{asymptotic} probability of majority illusions in the graphs 
by considering the neighborhood of a typical node in the limit graph. This works as follows:
First, we show that a sequence $(\boldsymbol{G}_n)$ of random colored graphs converges locally to a rooted limit graph $(\boldsymbol{G},o)$ if $n\rightarrow\infty$. For many graph sequences, it is already known that they locally converge, so we only have to show that the sequence still converges if it is colored.
\ser{We can then study the limiting \emph{proportion} of vertices in $\boldsymbol{G}_n$ with a particular property by studying the \emph{probability} that the vertex $o$ in $\boldsymbol{G}$ has this property. The properties that help us understand \texttt{Mm} illusions are (i) `$o$ is red' and (ii) `$o$ has mostly blue neighbors'. 
To be precise, }if the probability that $o$ in $\boldsymbol{G}$ is red is larger (resp. smaller) than $\frac{1}{2}$, then the probability that a majority of nodes in $\boldsymbol{G}_n$ is red converges to 1 (resp. 0) if $n\rightarrow \infty$. Similarly, if the probability that $o$ 
has a majority of blue neighbors is larger than $\frac{1}{2}$, then the probability that a majority of nodes in $\boldsymbol{G}_n$ has a majority of blue neighbors goes to 1 if $n\rightarrow\infty$.

\subsection{The Independent Identical Coloring Case}\label{sec:independent_coloring}
Here, we prove that if vertices of a large graph are colored independently with a bias for either color, then the probability of a Majority-majority illusion goes to $0$  when $n\rightarrow\infty$, as long as 
the proportion of pairs of vertices that are at distance at most 2 from each other goes to 0 ($G_1$ in Figure \ref{fig:example2} is an example of a graph where most pairs of nodes have distance more than 2)
. Note that this condition includes all graph sequences that converge locally in probability since local convergence implies that in any graph in the sequence most vertices have $O(1)$ nodes at distance at most $2$. Therefore, as $n\rightarrow\infty$, for a random vertex, the proportion of nodes that are at distance at most 2 goes to 0.

\begin{theorem}\label{thm:iid_coloring}
    Let $(\mathbf{G}^{\gamma_p}_n)$ be a sequence of graphs of growing size where nodes are colored red independently with probability $p\ne 0.5$, and blue otherwise. Further, suppose that the proportion of pairs of vertices that are at distance at most $2$ from each other goes to $0$. Then the probability of $\mathbf{G}^{\gamma_p}_n$ being in \texttt{Mm} illusion goes to $0$. 
\end{theorem}
Note that the condition rules out pedantic examples like a star graph of size $n$, where the center vertex is the only neighbor of all other vertices, and its color determines the perceived majority. The probability that 
the central vertex does not follow the dominant color converges to $\min\{p,1-p\}$ and not to $1$.
\begin{proof}
    Assume without loss of generality that $p<1/2$ so that blue is the majority color with probability tending to $1$ by the law of large numbers. Write $m_n$ for the number of vertices in $\mathbf{G}^{\gamma_p}_n$. Let $\mathcal{R}_n=\{i:M^{N_i}= \text{red}\}$ be the set of vertices in $\mathbf{G}^{\gamma_p}_n$ that see a red majority and write $R_n$ for its size. We are done if we show that $\mathbb{P}(R_n>m_n/2)\to 0$.  
    
    We will first show that $\mathbb{E}[R_n]\le (1/2-\epsilon)m_n$ for some $\epsilon>0$ for all $n$. Linearity of expectation gives that 
    \[\mathbb{E}[R_n]=\sum_i\mathbb{P}(i\in \mathcal{R}_n)=\sum_i\mathbb{P}\left(\operatorname{Bin}(\deg(i),p)>\deg(i)/2\right).\] Write $p_k=\mathbb{P}(\operatorname{Bin}(k,p)>k/2)$. 
    We will show that there is an $\epsilon>0$ such that $p_k<1/2-\epsilon$ for all $k\ge 0$. Duality of the binomial random variable implies that $p_k<1/2$ for all $k$, and to obtain uniformity we see that, by the Chernoff bound, for $k> K=\lfloor (1/2-p)^{-1}\rfloor$, $p_k<e^{-2}$. Thus, the claim follows with $\epsilon=1/2-\max\{e^{-2},p_0,\dots,p_K\}$.  

   \ser{We will use Chebyshev's inequality  to show that $\mathbb{P}(R_n>m_n/2)\to 0$. To that end, we need to bound the variance of $R_n$.} Observe that 
    \begin{align*}\mathbb{E}[R_n^2]&=\sum_{i,j}\mathbb{P}(i,j\in \mathcal{R}_n)\\&\le \sum_{i,j :N_i\cap N_j=\emptyset}\mathbb{P}(i,j\in \mathcal{R}_n)+\#\{i,j:N_i\cap N_j\ne \emptyset\}.\end{align*}
    By assumption, the second term is $o(m_n^2)$, and if $N_i\cap N_j=\emptyset$ then it is independent whether $i$ and $j$ are in $\mathcal{R}_n$, so \[\mathbb{P}(i,j\in \mathcal{R}_n)=\mathbb{P}(i\in \mathcal{R}_n)\mathbb{P}(j\in \mathcal{R}_n).\] Therefore, 
    \[\mathbb{E}[R_n^2]\le \sum_{i,j}\mathbb{P}(i\in \mathcal{R}_n)\mathbb{P}(j\in \mathcal{R}_n)+o(m_n^2)=\mathbb{E}[R_n]^2+o(m_n^2),\]
    so $\operatorname{Var}(R_n)=o(m_n^2)$.  We conclude, by Chebyshev's inequality, that 
    \[\mathbb{P}(R_n>m_n/2)\le \mathbb{P}(|R_n-\mathbb{E}[R_n]|>\epsilon m_n)\le \frac{\operatorname{Var}(Rn)}{\epsilon^2 m_n^2}\to 0\]
    as $n\to \infty$, which proves the theorem. 
\end{proof}

\begin{remark}
    Note that the above theorem does not only hold for a coloring where nodes are all blue or red with the \emph{same} probability $p$, but also for any graph coloring $\gamma_{X}$ that colors every node $i$ independently with a probability $p_i$ that is drawn randomly from a distribution $X$ of which the mean is not equal to 0.5. The reason is that, although the probability of being red is not exactly the same value for every node, these probabilities do come from the same distribution. Therefore, such coloring is an i.i.d. coloring as well, since for every node $i$ independently, $\mathbb{P}(i\text{ is red})=E[P(i\text{ is red}|X)]=\mathbb{E}[X]$.
\end{remark}

Since for all graph sequences that locally converge it holds that the proportion of  pairs of vertices are at distance $\leq2$ from each other goes to 0, a direct corollary is the following:
\begin{corollary}\label{cor:iid_coloring}
    Let $(\mathbf{G}^{\gamma_p}_n)$ be a locally converging sequence of graphs of growing size (for example the sequence of Erd\H{o}s--R\'{e}nyi graphs or the sequence of CSFP graphs) where nodes are colored red independently with probability $p\ne 0.5$, and blue otherwise. Then the probability of $\mathbf{G}^{\gamma_p}_n$ being in \texttt{Mm} illusion goes to $0$. 
\end{corollary}

\subsection{The Degree-Based Coloring Case}\label{sec:degree-based_coloring}
We construct colorings based on degree: let $\gamma_{d,\pi}$ for $\pi=(\pi_d)_{d\ge 0}$ be a sequence of distributions indexed by possible values of the degree $d$.
Since such coloring is determined independently for every node, 
using only information on the $1$-neighborhood, local convergence in probability of a growing sequence of graphs is enough to also obtain local convergence of the colored graphs.  This is the content of the next proposition. 

\begin{proposition}\label{prop:degree_colored_local}
    Let $(\boldsymbol{G}_n)$ be a sequence of graphs where $\boldsymbol{G}_n$ has $n$ vertices, 
    for which $(\boldsymbol{G}_n)\xrightarrow[locally]{p}(\boldsymbol{G},o)$. Then also $(\boldsymbol{G}_n^{\gamma_{d,\pi}})\xrightarrow[locally]{p}(\boldsymbol{G}^{\gamma_{d,\pi}},o)$.
\end{proposition}

\subsubsection{Erd\H{o}s--R\'{e}nyi graphs with degree-based colorings}
In Theorem \ref{thm:conv_ER}, we saw that Erd\H{o}s--R\'{e}nyi graphs converge locally in probability to Bienaymé--Galton--Watson trees with a $\operatorname{Poisson}$ offspring distribution. We will use this fact to show that on large Erd\H{o}s--R\'{e}nyi graphs where the coloring is correlated with the nodes' degree, majority illusions are very unlikely. 

We choose a coloring where nodes with higher degree have higher probability of being red. We could use the basic logistic growth function as a probability function, but since it has a left asymptote, we choose to use the cumulative distribution function of a Poisson distribution instead:  $\pi_d= CDF_{\operatorname{Poisson}(x_0)}(d)$\footnote{To prevent confusion: this choice of a \emph{Poisson distribution} is \textbf{not} related to the fact that the graphs converge to a \emph{Poisson tree}.\\
To get a more intuitive idea of how this coloring works, consider the following thought experiment: let a node start blue and count its neighbors one by one. For every counted neighbor, there is a fixed, independent probability that the node turns red, and otherwise it stays blue. The color that the node ends up with after counting all its neighbors is the color it gets according to $\pi_d$. }. This function starts at 0, 
so nodes with degree 0 are blue.
The Poisson distribution has a parameter $x_0$, which determines the expected value of the minimal number of neighbors with which a node will be red, and therefore determines the steepness of the curve. See Figure \ref{fig:Poisson_examples} in Appendix \ref{app:PoissonCDFs} for some examples of CDFs of Poisson distributions with different values of $x_0$.
We consider $x_0$ a parameter that can be tuned, so we do not choose one specific value for it.
A Corollary of Proposition \ref{prop:degree_colored_local} is  that the convergence in Theorem \ref{thm:conv_ER} also holds for $\gamma_{d,\pi}$-colored  Erd\H{o}s--R\'{e}nyi graphs $\boldsymbol{G}^{\gamma_{d,\pi}}_{n, \frac{c}{n}}$:
\begin{corollary}\label{cor:ER-TPoisson_y_d,cdfP(x_0),x}

     Let $(\boldsymbol{G}^{\gamma_{d,\pi}}_{n, \frac{c}{n}})$ with $\boldsymbol{G}^{\gamma_{d,\pi}}_{n, \frac{c}{n}} = \langle V_n, E_n, \gamma_{d,\pi, n}\rangle$ be a sequence of colored Erd\H{o}s--R\'{e}nyi graphs and $(T^{\gamma_{d,\pi}'}_{\operatorname{Poisson}(c)},o)$ a  colored rooted tree with a Poisson offspring such that $\gamma_{d,\pi, n}:V_n\rightarrow\{red, blue\}$ and $\gamma_{d,\pi}':V'\rightarrow\{red, blue\}$ ($V'$ is the set of nodes of $T_{\operatorname{Poisson}(c)}$) are colorings such that\\ $\gamma_{d,\pi, n}(v), \gamma_{d,\pi}'(v)\mapsto \begin{cases}
         red\text{ with probability } CDF_{\operatorname{Poisson}(x_0)}(d_v)\\
         blue\text{ with probability } 1 - CDF_{\operatorname{Poisson}(x_0)}(d_v)
     \end{cases}$
     Then, 
     $$(\boldsymbol{G}^{\gamma_{d,\pi}}_{n, \frac{c}{n}})\xrightarrow[locally]{p}(T^{\gamma_{d,\pi}'}_{\operatorname{Poisson}(c)},o).$$
\end{corollary}

Using this Corollary, we will show that the probability of large Erd\H{o}s--R\'{e}nyi graphs colored according to $\gamma_{d,\pi}$ being in \texttt{Mm} illusion goes to 0.
\begin{theorem}\label{thm:conv_ER_coloring_degree}
    Let $c$ and $x_0>c$ be given integers.
    Let $(\boldsymbol{G}^{\gamma_{d,\pi}}_{n, \frac{c}{n}})$ be a sequence of Erd\H{o}s--R\'{e}nyi graphs where a node $v$ with degree $d_v$ is colored red with probability $\mathbb{P}(red)=CDF_{\operatorname{\operatorname{Poisson}}(x_0)}(d_v)$, and blue otherwise. If $n\rightarrow \infty$, the probability of $\boldsymbol{G}^{\gamma_{d,\pi}}_{n, \frac{c}{n}}$ being in \texttt{Mm} illusion goes to 0.
\end{theorem}

\subsubsection{Continuum scale-free percolation graphs with degree-based colorings}
In \cite{van_der_Hofstad_2023} it is shown that a sequence of graphs of the CSFP model converges locally in probability, as $n\rightarrow\infty$, to an infinite \textit{rooted} CSFP (Corollary 4 of their paper)\ser{, so we may apply our Proposition~\ref{prop:degree_colored_local} to this model.}

\ser{To illustrate that majority illusions may occur, and that degree inhomogeneity is 
one of the underlying reasons for this, we take a very simple coloring $\gamma_k$ based on degree:} If a node has more than $k$ neighbors, it is red, else it is blue\footnote{We could use the probabilistic Poisson function that we used for Erd\H{o}s--R\'{e}nyi graphs, which is a softer version of the simple threshold function, also here. However, to make the proof better readable we chose to simplify the coloring function here. Note that the intuition behind the two colorings is the same.}. 
Informally, high-degree nodes are the neighbor of many, and are therefore overrepresented in the `observed majority', giving fertile ground for majority illusions to occur. (This is directly related to the intuition behind the friendship paradox: typically, our friends have more friends than we do.) We prove that majority illusions occur with high probability when $k=10, d=2, \alpha=4, \lambda=1/(10201\pi^3), \beta=0.505$. Of course this parameter choice is rather arbitrary, but it allows us to show with a relatively simple calculation that in \ser{random graph models that incorporate several empirically observed features of real-world networks (heterogeneous degree distributions, spatial embedding, clustering and long-range connections), majority illusions can be highly probable.} 


\begin{theorem}\label{thm:conv_CSFP}
    Let $(\boldsymbol{G}^{\gamma}_{n})$ 
    be a sequence of continuum scale-free percolation graphs with parametrization as above, where a node $v$ with degree $d_v$ is colored red if $d_v>10$ and blue otherwise. If $n\rightarrow \infty$, the probability of $(\boldsymbol{G}^{\gamma}_{n},o)$ being in \texttt{Mm} illusion goes to 1.
\end{theorem}
\ser{With Proposition~\ref{prop:degree_colored_local} in hand, we need to perform our calculations on the limit object. The limit is constructed by taking a Poisson process on $\mathbb{R}^d$, planting an extra point in the origin, giving each of the infinitely many points an independent Pareto weight, and connecting points with an edge according to the same rule as in the finite model. Then, the planted vertex serves as the root.} 
A core idea in the proof is to consider two different independent point processes: first the one defining the neighbors of the root, and then the one defining the neighbors of a neighbor of the root that are not neighbors of the root themselves. \ser{We use the first point process} to show that the probability that the root has high degree (`is red') is low, and then \ser{the second point process shows} that the probability that its neighbors have high degree is high \ser{(even when we ignore edges between root neighbors)}. 
\begin{proof}
By Corollary 4 of \cite{van_der_Hofstad_2023}, a sequence of CSFP graphs converges locally in probability to an infinite rooted CSFP graph when $n\rightarrow\infty$. Then, by Proposition \ref{prop:degree_colored_local}, the sequence 
$(G^\gamma_n)$ of colored CSFP graphs converges locally in probability to the rooted $\gamma_k$-colored CSFP 
$(G^\gamma_n, o)$.

    We will show that with our parameters, the probability that the root $o$ is red is less than half: $\mathbb{P}(D_o>k) <\frac{1}{2}$, but the probability that the majority of neighbors of the root is red is more than half: $\mathbb{P}(\#\{i\in N_o:d_i>k\}>\frac{D_0}{2})>\frac{1}{2}$. If this is the case, the probability that a strict majority of the nodes is blue will converge to 1, but the probability that a majority of nodes sees a red majority converges to 1 as well. Therefore, in the limit, the probability that the graph is under \texttt{Mm} illusion converges to 1 as well.
    \begin{lemma}\label{lem:exp_deg}
        If $\beta>d$, then, the expected degree of a vertex $i$ given that its weight $W_i = w$ is 
        \begin{equation*}
            \mathbb{E}[D_i|W_i=w] = v_d\Gamma\left(1-\frac{d}{\alpha}\right)\lambda^\frac{d}{\alpha}\mathbb{E}[W^\frac{d}{\alpha}]w^\frac{d}{\alpha},
        \end{equation*}
        where $d$ is the dimension of the space, $v_d$ is the volume of a unit ball (so $v_2=\pi$), $\lambda$ and $\alpha$ are parameters of the Poisson point process defining the graph.
    \end{lemma}
        \begin{proof}[proof of lemma]
        For $z\in \mathbb{R}^d$, we see that \[p_w(z)=\mathbb{E}\left[1-\exp\left(-\frac{\lambda w W}{||z||^\alpha}\right)\right]
        \] is the probability that a vertex in position $x$ with weight $w$ connects to a vertex in position $x+z$ if both vertices are in the point process. Here, we have taken the expectation over the possible weight $W$ of the second vertex.  Thus, by linearity of expectation, integrating over all possible positions $z$ of a possible neighbor, we get that 
        \begin{align*}\mathbb{E}[D_i|W_i=w]=\int_{\mathbb{R}^d}p_w(z)dz=\int_{ \mathbb{R}^d} \mathbb{E}\left[1-\exp\left(-\frac{\lambda w W}{||z||^\alpha}\right)\right] dz.\end{align*}
        Then a standard argument that consists of swapping the integral and the expectation (using Tonelli), switching to polar coordinates and then using the definition of the Gamma function gives the result.
        \end{proof}

    Thus, for our parameter choice, according to Lemma \ref{lem:exp_deg}, the expected degree of a node $i$ in the CSFP with weight $w$ is
    \begin{eqnarray*}
        \mathbb{E}[D_i|W_i=w] &=& \pi\Gamma\left(\frac{1}{2}\right)\lambda^\frac{1}{2}\mathbb{E}[W^\frac{1}{2}]w^\frac{1}{2}\\
        &=&\pi\sqrt{\pi}\sqrt{\lambda}\mathbb{E}[\sqrt{W}]\sqrt{w}
    \end{eqnarray*}
    \ser{An easy calculation shows that }for a Pareto random variable $W$ with parameter $\beta$, $\mathbb{E}[W^q]=\frac{\beta}{\beta - q}$. Thus, since we chose $\beta = 0.505$, we have that $\mathbb{E}[\sqrt{W}]=\frac{0.505}{0.005}=101$. We chose $\lambda = \frac{1}{10201\pi^3} =\frac{1}{101^2\pi^3}$ so that we arrive at the easy expression:
    \begin{eqnarray*}
        \mathbb{E}[D_i|W_i=w] &=& \sqrt{w}
    \end{eqnarray*}
    
    This means that the expected degree of a node scales with the square root of its weight.

We consider the root $o$ with weight $W_o$. Its expected degree is $\mathbb{E}[D_o\mid W_0]=\sqrt{W_o}$. The model has the property that the root's degree is a Poisson random variable with mean equal to its expected degree: \begin{eqnarray*}
    D_o|\sqrt{W_o}=x \sim \operatorname{Poisson}(x)
\end{eqnarray*}
Also, the probability that the expected degree of the root is greater than some value $x$ is given by  $$\mathbb{P}(\sqrt{W_o}>x)=\mathbb{P}(W_o>x^2)=(x^2)^{-\beta}=x^{-1.01}.$$
Now\ser{, using the law of total probability,} we split the event the root is red (that it has degree $D_o>k=10$) into the weight being either smaller or larger than 25:
\begin{eqnarray*}
    \mathbb{P}(D_o>10)&=& \ser{\mathbb{P}(D_o>10, \sqrt{W_o}>5)+\mathbb{P}(D_o>10, \sqrt{W_o}\le 5)}\\
    &\leq& \mathbb{P}(\sqrt{W_o}>5)+\mathbb{P}(\operatorname{Poisson}(5)>10)\\
    &\leq& 5^{-1.01}+0.01370<0.5
\end{eqnarray*}

Now we consider the neighbors of the root. \ser{Their weight is biased}: because we know they are neighbors, it is more likely for them to have a high weight.  This is illustrated by the following lemma. 

\begin{lemma}
    If $\beta>d/\alpha$, then the weight of a neighbor of the root is distributed as $\mathbb{P}(U>u)=u^{-\beta+d/\alpha}.$
\end{lemma}
\begin{proof}[proof of lemma]
    If the root has weight $w$, then a potential vertex in $z\in \mathbb{R}^d$ with weight $u$ connects to the root with probability $1-\exp(-\frac{\lambda w u}{||z||^\alpha})$. Therefore, the neighbors of the root and their weights are a thinned Poisson process on $\mathbb{R}^d\times[1,\infty)$ with intensity \[(1-\exp(-\frac{\lambda w u}{||z||^\alpha}))\beta u^{-\beta - 1} dudz.\] (Indeed, $1\times \beta u^{-\beta - 1}$ is the unfiltered intensity, and the extra factor accounts for the criterion of connecting to the root.) We are not interested in the location of the neighbors of the root, so we integrate this out, to obtain that the intensity of root neighbors with weight in $du$ is 
    \[ \int_{\mathbb{R}^d}(1-\exp(-\frac{\lambda w u}{||z||^\alpha})dz\beta u^{-\beta - 1}.\]
    The change of variable $z=(\lambda w u)^{1/\alpha}y$ allows us to compute that this is proportional to $u^{-\beta+d/\alpha - 1}$ for all $u\in (1,\infty)$, so the only possible distribution is $\mathbb{P}(U>u)=u^{-\beta+d/\alpha}$ as claimed. 
\end{proof}

Thus, while the original Pareto distribution of the weights $W$ satisfies that $\mathbb{P}(W>w)=w^{-0.505}$, the weight of the neighbors of the root $U$ is much more heavy-tailed and satisfies that $\mathbb{P}(U>u)=u^{-\beta +\frac{1}{2}}=u^{-0.005}$.

Just as with the root itself, the expected value of the degree of a neighbor $i$ of the root, given that its weight is $u$, is $$\mathbb{E}(D_i|U_i=u)=\sqrt{u}+1,$$
\ser{where we add $1$ to account for the known connection to the root, using that a Poisson random variable conditioned to be at least $1$ has the law of the same Poisson random variable plus $1$. Using the distribution for $U_i$,} $$\mathbb{P}(\sqrt{U_i}>y) = \mathbb{P}(U_i>y^2)=y^{2(-0.005)}=y^{-0.01}.$$
Now, let $D_i^{out}$ be the `out-degree' of a neighbor $i$ of the root, that is, the number of edges from $i$ to nodes that are not \ser{the root or} neighbors of the root $o$. We want to show that, with high probability, $D_i^{out}>10$ for most neighbors of the root, which means that also $D_i>10$ for most neighbors, implying that most neighbors are red.

Since we consider only neighbors of $i$ that are not neighbors of $o$, $D_i^{out}$ is determined by an independent Poisson distribution, with mean $\mu_i$. Since the expected \emph{total} degree of $i$ is given by $\sqrt{U_i}+1$, and this includes \ser{the root and (in expectation)} maximally $\sqrt{W_0}$ neighbors of the root, we have that $$\mu_i\geq \sqrt{U_i}-\sqrt{W_0}.$$

Now, if it is the case that $\sqrt{W_o}\leq 100$ and $\sqrt{U_i}\geq 120$, then $\mu_i\geq20$, so $D_i^{out}\succeq\operatorname{Poisson}(20)$. 
For the weight of the neighbors to be large enough, we know that $$\mathbb{P}(\sqrt{U_i}\geq 120)=120^{-0.01}\approx0.953.$$
We will later calculate the probability that the root's weight $W_o$ is indeed small enough, but given that $\sqrt{W_o}\leq 100$, we now have that
\begin{eqnarray*}
p_{red}&:=&\mathbb{P}\left(D_i^{out}>10|\sqrt{W_o}\leq100\right)\\
&\geq&\mathbb{P}(\sqrt{U_i}>120)\mathbb{P}(\operatorname{Poisson}(20)>10)\\
    &\approx& 0.953\cdot 0.989\approx 0.943
\end{eqnarray*}
Let $B:=\#\{i\in N_o:D_i^{out}\leq10\}$ be the number of neighbors of the root with out-degree at most 10. Each neighbor has probability at most $1-p_{red}\approx0.057$ of belonging to $B$, so $$\mathbb{E}[B|\sqrt{W_o}\leq100, D_o=d_o]\leq d_o(1-p_{red}),$$
where $d_o\geq 1$ is the degree of the root. (Observe that it is \emph{not} independent whether different neighbors of the root are in $B$ because they rely on the same set of weighted vertices to connect to. However, this is not needed when using linearity of expectation.)
 Our goal is to show that it is likely that a majority of the neighbors has degree larger than 10, i.e. that  $B< \lceil\frac{d_o}{2}\rceil$.
We will use \ser{conditional} Markov's inequality, \ser{which says that for a non-negative random variable $X$ and an event $\mathcal{E}$},  $\mathbb{P}(X\geq a)\leq\frac{\mathbb{E}[X\mid \mathcal{E}]}{a}$. Therefore, we have that \begin{eqnarray*}
    \mathbb{P}(B\geq d_0/2\mid\sqrt{W_o}\leq100, D_o=d_o)&\leq& \frac{\mathbb{E}[B\ser{\mid \sqrt{W_o}\leq100, D_o=d_o}]}{d_0/2}\\
    &\leq&\frac{d_o(1-p_{red})}{d_0/2}=2(1-p_{red}).
\end{eqnarray*}
Hence, the probability that a majority of the neighbors is red is
\begin{eqnarray}\label{eq:out_deg_cond_Wo_Do}
    \mathbb{P}(B<\lceil\frac{d_o}{2}\rceil\mid \sqrt{W_o}\leq100, D_o=d_o)\geq1-2(1-p_{red})\approx 0.886.
\end{eqnarray}
We still condition on the probability that the weight of the root is small enough and that the root has at least one neighbor ($D_o\geq1$), but we can calculate these probabilities:
\begin{eqnarray}\label{eq:Wo}
    \mathbb{P}(\sqrt{W_o}\leq100)=1-\mathbb{P}(\sqrt{W_o}>100)=1-100^{-1.01}\approx0.99045,
\end{eqnarray}
and given that $\sqrt{W_o}=x$, $D_o\sim\operatorname{Poisson}(x)$, so since $\sqrt{W_o}\geq1$ \ser{because $W$ only puts mass on $(1,\infty)$}
\begin{eqnarray}\label{eq:Do}
    \mathbb{P}(D_o\geq1|\sqrt{W_o}=x)=1-e^{-x}\geq1-e^{-1}\approx0.632.
\end{eqnarray}
Hence, combining \eqref{eq:out_deg_cond_Wo_Do}, \eqref{eq:Wo}, and \eqref{eq:Do}, we have:
$$\mathbb{P}(B<\lceil\frac{D_o}{2}\rceil)\ser{\gtrsim} 0.886\cdot0.99045\cdot0.632 =0.555>\frac{1}{2}.$$
\ser{This concludes the proof. }
 
\end{proof}


Hence, with a coloring based on the degree of the nodes, as $n\rightarrow\infty$, the probability of an Erd\H{o}s-R\'{e}nyi graph being in Majority-majority illusion goes to 0, but for a CSFP graph with the given parameters the probability goes to 1.

\section{Discussion and Conclusions}
We derived 0-1 laws for the probability of the occurrence of majority illusions on large random networks.
By doing so, we  showcased the use of local convergence, 
and therefore of probability theory, for the study of phenomena on (social) networks. Hereby we added an approach to the study of majority illusions, that so far was done mainly experimentally and graph-theoretically.

The results are summarized in Table \ref{tab:results}. We found that under independent random colorings, on any graph sequence that locally converges (and even a superclass of locally converging graph sequences) the probability of a graph being in Majority-majority illusion goes to 0 as the size $n$ of the graph goes to $\infty$.
Furthermore, we saw that under a coloring based on the degree of nodes, 
for Erd\H{o}s--R\'{e}nyi graphs, this probability goes to 0 as the size of the graph $n$ goes to $\infty$.
By contrast, in a continuum scale-free percolation model 
under a coloring that depends on the degree of the nodes, 
it goes to 1.
Our results here can be seen as elementary examples of the proof strategy: 
We can use our technique for any coloring that is based only on local properties of the nodes, and any random graph with a local limit. 
As a first step, 
a more natural coloring choice in a social network would be a coloring based on homophily (see e.g. \cite{McPherson2001
}).
In such colorings, nodes are more likely to have neighbors of their own color 
than 
of different colors. 
In that direction, in \cite{Ramanan2021}[Proposition 1.16], the authors show that if a converging graph sequence is colored according to a so-called random \emph{Gibbs} probability measure, the sequence  converges 
locally 
to a limit graph that is colored with the same 
measure. Since 
such
measures can take into account the colors of neighboring nodes, we could define 
one that lets the graph exhibit homophily, and conclude 
convergence.

For the continuum scale-free percolation model we used an LLM to help us find suitable parameters with which the probability of Majority-majority illusions converges to 1. We leave for future research the problem of finding the entire parameter range for which this is the case.



\balance


\bibliography{AAMAS/bibliography}

\newpage
\newpage
\appendix
\section{Omitted proofs and auxiliary material}\label{app:aux/proofs}

\subsection{Proof of Proposition \ref{prop:degree_colored_local}}
\noindent\textbf{Proposition \ref{prop:degree_colored_local}.} \emph{Let $(\boldsymbol{G}_n)$ be a sequence of graphs where $G_n$ has $n$ vertices, for which $(\boldsymbol{G}_n)\xrightarrow[locally]{p}(\boldsymbol{G},o)$. Then also $(\boldsymbol{G}_n^{\gamma_{d,\pi}})\xrightarrow[locally]{p}(\boldsymbol{G}^{\gamma_{d,\pi}},o)$.}
\begin{proof}

    Let $(G_n)$ be a sequence of graphs where $G_n$ has $n$ vertices and $n\to\infty$ for which it holds that $(G_n)\xrightarrow[locally]{p}(G,0)$. We need to show that also $(G_n^{\gamma_{d,\pi}})\xrightarrow[locally]{p}(G^{\gamma_{d,\pi}},0)$.

Let $(G^\gamma,0)$ be a rooted colored locally finite graph and let $k\ge 0$. It is sufficient to show that for $U_n$ and $U'_n$ two independent samples from the vertex set of $G_n$, as $n\to \infty$,

\begin{align*}
&\mathbb{P}(B_k(G_n^{\gamma_{d,\pi}},U_n)\simeq (G^\gamma,0)) \to \mathbb{P}(B_k(G^{\gamma_{d,\pi}}, 0) \simeq (G^{\gamma}, 0))\\
\intertext{ and}
&\mathbb{P}\left (B_k(G_n^{\gamma_{d,\pi}},U_n)\simeq  (G^\gamma,0), B_k(G_n^{\gamma_{d,\pi}},U'_n)\simeq  (G^\gamma,0))\right)\\
 &  \qquad \to \mathbb{P}(B_k(G^{\gamma_{d,\pi}}, 0) \simeq  (G^\gamma,0))^2. 
\end{align*}

For the first statement, observe that to sample the coloring ${\gamma_{d,\pi}}$  restricted to $B_k(G_n,U_n)$, we need to know the degrees of all vertices in $B_k(G_n,U_n)$, so we need to observe the $(k+1)st$ neighborhood of $U_n$ in $G_n$. Thus, we will need to consider all possible $(k+1)$-neighborhoods $(G',0)$ that realize $k$-neighborhood $(G,0)$. To be precise, let $\mathcal{G}'$ be the set of locally finite rooted graphs $(G',0)$ for which $B_{k+1}(G',0)\simeq (G,0)$.  Then, 
\begin{equation*}
\begin{aligned}
&\mathbb{P}(B_k(G_n^{\gamma_{d,\pi}},U_n)\simeq (G^\gamma,0))\\
 &=\sum_{(G',0)\in \mathcal{G}'} \mathbb{P}( B_k(G_n^{\gamma_{d,\pi}},U_n)\simeq (G^\gamma,0), B_{k+1}(G_n,U_n)\simeq (G',0))\\
&= \sum_{(G',0)\in \mathcal{G}'} \mathbb{P}(B_{k+1}(G_n,U_n)\simeq (G',0))  \mathbb{P}(B_k((G')^{\gamma_{d,\pi}}, 0)=(G^\gamma,0))\\
&\to \sum_{(G',0)\in \mathcal{G}'} \mathbb{P}(B_{k+1}(G,0)\simeq (G',0)) \mathbb{P}(B_k((G')^{\gamma_{d,\pi}}, 0)=(G^\gamma,0))\\
&=\mathbb{P}(B_k(G^{\gamma_{d,\pi}}, 0) \simeq (G^{\gamma}, 0))
\end{aligned}
\end{equation*}
where the convergence follows from $(G_n)\xrightarrow[locally]{p}(G,0)$. 

For the second statement, we need to consider the correlation of $B_k(G_n^{\gamma_{d,\pi}},U_n)$ and $B_k(G_n^{\gamma_{d,\pi}},U'_n)$. We observe that the coloring of $B_k(G_n,U_n)$ only depends on $B_{k+1}(G_n,U_n)$ and the coloring of $B_k(G_n,U'_n)$ only depends on $B_{k+1}(G_n,U'_n)$, so if $B_{k+1}(G_n,U_n)$ and $B_{k+1}(G_n,U'_n)$ do not overlap, then the colorings are independent conditional on $B_{k+1}(G_n,U_n)$ and $B_{k+1}(G_n,U'_n)$. We will first show that it is very unlikely that $B_{k+1}(G_n,U_n)$ and $B_{k+1}(G_n,U'_n)$ overlap. Observe that this bad event happens if and only if $U'_n\in B_{2k+1}(G_n,U_n)$, so, since $U'_n$ is a uniformly chosen vertex, this event has probability 
\[\frac{|B_{2k+1}(G_n,U_n)| }{|G_n| } .\] 
We show that it is unlikely that this ratio is large.
Fix $\varepsilon>0$. Let $K$ be large enough that $\mathbb{P}(|B_{2k+1}(G,0)|>K)<\varepsilon/4$. Such a $K$ exists because $(G,0)$ is locally finite. By $G_n\xrightarrow[locally]{p}(G,0)$ we know that $\mathbb{P}(|B_{2k+1}(G_n,U_n)|>K) \to\mathbb{P}(|B_{2k+1}(G,0)|>K)$, so for $n$ large enough $\mathbb{P}(|B_{2k+1}(G_n,U_n)|>K)<\varepsilon/2$. Then, because $|G_n|\to \infty$, we may pick $n$ even larger so that $K/|G_n| <\varepsilon/2 $, so that for $n$ large enough

\begin{align*}
    \mathbb{P}&(B_{k+1}(G_n,U_n)\cap B_{k+1}(G_n,U'_n) \neq \emptyset) \\
&\le \frac{K }{|G_n| }+ \mathbb{P}(|B_{2k+1}(G_n,U_n)|>K) \le \varepsilon .
\end{align*} 
Then, by similar reasoning as before, 
\begin{equation*}
    \begin{aligned}
    & \left| \mathbb{P}(B_k(G_n^{\gamma_{d,\pi}},U_n)\simeq  (G^\gamma,0), B_k(G_n^{\gamma_{d,\pi}},U'_n)\simeq  (G^\gamma,0))) \right. \\
    &- \sum_{(G_1,0),(G_2,0)\in \mathcal{G}} \mathbb{P}(B_{k+1}(G_n,U_n)\simeq (G_1,0), B_{k+1}(G_n,U_n)\simeq (G_2,0)) \\
    &\qquad  \left. \times \mathbb{P}(B_k(G_1^{\gamma_{d,\pi}}, 0)=(G^\gamma,0))\mathbb{P}(B_k(G_2^{\gamma_{d,\pi}}, 0)=(G^\gamma,0))\right|\\
    &\le \mathbb{P}\left(B_{k+1}(G_n,U_n)\cap B_{k+1}(G_n,U'_n) \neq \emptyset\right )\le \varepsilon. 
    \end{aligned}
\end{equation*}

But, since $(G_n^{\gamma_{d,\pi}})\xrightarrow[locally]{p}(G^{\gamma_{d,\pi}},0)$, 

\begin{equation*}
    \begin{aligned} 
&\sum_{(G_1,0),(G_2,0)\in \mathcal{G}} \mathbb{P}(B_{k+1}(G_n,U_n)\simeq (G_1,0), B_{k+1}(G_n,U_n)\simeq (G_2,0))\\
&\qquad\qquad \times \mathbb{P}(B_k(G_1^{\gamma_{d,\pi}}, 0)=(G^\gamma,0))\mathbb{P}(B_k(G_2^{\gamma_{d,\pi}}, 0)=(G^\gamma,0))\\
&\to  \sum_{(G_1,0),(G_2,0)\in \mathcal{G}} \mathbb{P}(B_{k+1}(G,0)\simeq (G_1,0)) \mathbb{P}(B_{k+1}(G,0)\simeq (G_2,0)) \\
&\qquad\qquad \times \mathbb{P}(B_k(G_1^{\gamma_{d,\pi}}, 0)=(G^\gamma,0))\mathbb{P}(B_k(G_2^{\gamma_{d,\pi}}, 0)=(G^\gamma,0))\\
&=\mathbb{P}(B_k(G^{\gamma_{d,\pi}}, 0) \simeq  (G^\gamma,0))^2 ,
\end{aligned}
\end{equation*}

which implies the second statement because $\varepsilon$ was arbitrary, so that the proposition follows. 
\end{proof}

\subsection{Proof of Theorem \ref{thm:conv_ER_coloring_degree}}
\noindent\textbf{Theorem \ref{thm:conv_ER_coloring_degree}.} \emph{    Let $c$ and $x_0>c$ be given integers.
    Let $(\boldsymbol{G}^{\gamma_{d,\pi}}_{n, \frac{c}{n}})$ be a sequence of Erd\H{o}s--R\'{e}nyi graphs where a node $v$ with degree $d_v$ is colored red with probability $\mathbb{P}(red)=CDF_{\operatorname{Poisson}(x_0)}(d_v)$, and blue otherwise. If $n\rightarrow \infty$, the probability of $\boldsymbol{G}^{\gamma_{d,\pi}}_{n, \frac{c}{n}}$ being in \texttt{Mm} illusion goes to 0.}
\begin{proof}
To compute the probability that the $n^{th}$ Erd\H{o}s--R\'{e}nyi graph is in \texttt{Mm} illusion, we need to compute the probability that a node in the graph has a majority of red neighbors and the probability that a node itself is red. 
By Corollary \ref{cor:ER-TPoisson_y_d,cdfP(x_0),x}, we have that the probability that a node has a majority of red neighbors in $\boldsymbol{G}^{\gamma_{d,\pi}}_{n, \frac{c}{n}}$ converges to the probability that a node has a majority of red neighbors in the Poisson tree:
\begin{equation*}
\frac{1}{n}\sum_{v\in \boldsymbol{G}^{\gamma_{d,\pi}}_{n, \frac{c}{n}}}\mathbbm{1}_{\{v \text{ sees red majority}\}} 
\xrightarrow{p}\mathbb{P}(o \text{ sees red majority in }T^{\gamma'_{d,\pi}}_{\operatorname{Poisson}(c)}) 
\end{equation*}
and the probability that a node is red in $\boldsymbol{G}^{\gamma_{d,\pi}}_{n, \frac{c}{n}}$ converges to the probability that a node is red in the Poisson tree:
\begin{equation*}
    \begin{aligned}
        \frac{1}{n}\sum_{v\in \boldsymbol{G}^{\gamma_{d,\pi}}_{n, \frac{c}{n}}}\mathbbm{1}_{\{v \text{ is red}\}}\xrightarrow{p}\mathbb{P}(o \text{ is red in } T^{\gamma'_{d,\pi}}_{\operatorname{Poisson}(c)}). \\
    \end{aligned}
\end{equation*}
To compute $\mathbb{P}(o \text{ sees red majority in }T^{\gamma'_{d,\pi}}_{\operatorname{Poisson}(c)}) $, we need the probabilities that the neighbors of $o$ are red. For a node $v$ with degree $d_v$ we have  that the probability that $v$ is red in $(T^{\gamma_{d,\pi}'}_{\operatorname{Poisson}(c)},o)$ is $\mathbb{P}(v\text{ is red})=CDF_{\operatorname{Poisson}(x_0)}(d_v)$.
Since we do not know the exact degrees of all nodes in $T^{\gamma_{d,\pi}}_{\operatorname{Poisson}(c)}$, but we know that they are distributed according to $\operatorname{Poisson}(c)$, the probability that a node is red if we do not know the degree of the node is $CDF_{\operatorname{Poisson}(x_0)}(\mathbf{P})$, where $\mathbf{P}\sim \operatorname{Poisson}(c)$. 
Therefore, the probability that a majority of neighbors of $o$ are red in $T^{\gamma'_{d,\pi}}_{\operatorname{Poisson}(c)}$ is
\begin{equation}\label{eq:o_sees_red_majority_ydegree}
    \begin{aligned}
        \mathbb{P}&(o\text{ sees red majority})\\
        =&\sum_{d=1}^{\infty}\left(\mathbb{P}(\operatorname{Poisson}(c)=d)\cdot\sum_{j=\lceil\frac{d+1}{2}\rceil}^d \mathbb{P}(j\text{ neighbors are red})\right) \\
        =&\sum_{d=1}^{\infty}\left(\mathbb{P}(\operatorname{Poisson}(c)=d)\cdot  \sum_{j=\lceil\frac{d+1}{2}\rceil}^d \mathbb{P}(\operatorname{Bin}(d, CDF_{\operatorname{Poisson}(x_0)}(\mathbf{P}))=j)\right), 
    \end{aligned}
\end{equation}
where  $\mathbf{P}\sim \operatorname{Poisson}(c)$.

The probability that a node itself is red in the Poisson tree is $$\mathbb{P}(o \text{ is red in } T^{\gamma'_{d,\pi}}_{\operatorname{Poisson}(c)}) = CDF_{\operatorname{Poisson}(x_0)}(\mathbf{P})$$ with $\mathbf{P}\sim \operatorname{Poisson}(c)$.
By the tower principle of conditional expectation, we have that $\mathbb{E}[CDF_{\operatorname{Poisson}(x_0)}(\mathbf{P})]=\mathbb{E}[\mathbb{E}[CDF_{\operatorname{Poisson}(x_0)}(\mathbf{P})|\mathbf{P}]]$. Since $\mathbb{E}[\mathbf{P}]=\mathbb{E}[\operatorname{Poisson}(c)]=c$ and $CDF_{\operatorname{Poisson}(\lambda)}(k) > \frac{1}{2}$ if $k\geq \lambda$ and $CDF_{\operatorname{Poisson}(\lambda)}(k) < \frac{1}{2}$ if $k< \lambda$ \cite{Adell2005}
,
we have that 
\begin{equation*}
    \begin{aligned}
        \mathbb{E}[\mathbb{P}(o\text{ is red})]&=\\
        \mathbb{E}[\mathbb{E}[CDF_{\operatorname{Poisson}(x_0)}(\mathbf{P})|\mathbf{P}]]
        & \begin{cases}
                >\frac{1}{2} & \text{if } c\geq x_0\\
                <\frac{1}{2} & \text{if } c<x_0
            \end{cases}
    \end{aligned}
\end{equation*}

Therefore, the probability that a majority of nodes in the graph is red converges to 1 if $c\geq x_0$ and to 0 if $c<x_0$, and vice versa the probability that a majority of nodes in the graph is blue converges to 0 if $c\geq x_0$ and to 1 if $c<x_0$.

Therefore, given $c<x_0$, the probability that the graph is under \texttt{Mm} illusion is the probability that a majority of nodes sees a red majority, which converges as follows:
\begin{lemma}\label{lem:conv_Mm_ydegree}
     Given any $x_0$ such that $c<x_0$, the probability that $\boldsymbol{G}^{\gamma_{d,\pi}}_{n, \frac{c}{n}}$ is under \texttt{Mm} illusion converges to 1 (if $n\to\infty$) if the probability that o sees a red majority in $T^{\gamma'_{d,\pi}}_{\operatorname{Poisson}(c)}$ (Equation \ref{eq:o_sees_red_majority_ydegree}) is bigger than $\frac{1}{2}$
,
and to 0 if it is smaller than $\frac{1}{2}$.
\end{lemma}

We will show that, given $c<x_0$ and $p<\frac{1}{2}$,
\begin{align}\label{eq:red_maj}
    \mathbb{P}(o \text{ sees red maj. in }T^{\gamma'_{d,\pi}}_{\operatorname{Poisson}(c)}) 
     = \mathbb{P}(o \text{ sees red maj. in }T^{\gamma'_{p}}_{\operatorname{Poisson}(c)}).
\end{align}

The right hand side corresponds to the probability that $o$ sees a red majority in a Poisson tree with a simpler coloring: all nodes are colored red with probability $p$ and blue otherwise. This probability is given by
\begin{equation}\label{eq:v_sees_red_majority}
\begin{aligned}
& \mathbb{P}(o \text{ sees red majority in }T^{\gamma'_p}_{\operatorname{Poisson}(c)}) = \\
&\sum_{d=1}^{\infty}\left(\mathbb{P}(\operatorname{Poisson}(c)=d)\cdot\sum_{j=\lceil\frac{d+1}{2}\rceil}^d \mathbb{P}(\operatorname{Bin}(d, p)=j)\right) 
\end{aligned}
\end{equation}

 We show that this is always smaller than $\frac{1}{2}$:

\begin{lemma}\label{lem:v_sees_red_majority<0.5}
    $$
    \sum_{d=1}^{\infty}\left(\mathbb{P}(\operatorname{Poisson}(c)=d)\cdot\sum_{j=\lceil\frac{d+1}{2}\rceil}^d \mathbb{P}(\operatorname{Bin}(d, p)=j)\right)<\frac{1}{2}
    $$ for any $c\in \mathbb{N}$ and $p<\frac{1}{2}$.
\end{lemma}
\begin{proof}[Proof of Lemma \ref{lem:v_sees_red_majority<0.5}]
   We use two facts.
   \begin{enumerate}
       \item  $$\sum_{d=1}^{\infty}\left(\mathbb{P}(\operatorname{Poisson}(c)=d)\right)<1$$ for any $c\in \mathbb{N}$, since $\sum_{d=0}^{\infty}\left(\mathbb{P}(\operatorname{Poisson}(c)=d)\right)=1$ and $\mathbb{P}(\operatorname{Poisson}(c)=0)>0$ for any $c\in \mathbb{N}$.
       \item  For any $p<0.5$, for any $d\in \mathbb{N}$, $$\sum_{j=\lceil\frac{d+1}{2}\rceil}^d \mathbb{P}(\operatorname{Bin}(d, p)=j)<\frac{1}{2}.$$ Proof: For $p=0.5$, the distribution of $\operatorname{Bin}(d, p)$ would be symmetrical around $\frac{d}{2}$, so the sum $\sum_{j=\frac{d}{2}}^d \mathbb{P}(\operatorname{Bin}(d, p)=j)$ would be exactly $\frac{1}{2}$ (if $d$ is even so $j$ is integer, but this does not matter for the proof). However, here we take the sum from $j=\lceil\frac{d+1}{2}\rceil$ instead of from $j=\frac{d}{2}$ (and $\lceil\frac{d+1}{2}\rceil>\frac{d}{2}$), and we have $p<0.5$ instead of $p=0.5$, both of which make the total sum smaller. 
   \end{enumerate}
   Now we can write the sum $$\sum_{d=1}^{\infty}\left(\mathbb{P}(\operatorname{Poisson}(c)=d)\cdot\sum_{j=\lceil\frac{d+1}{2}\rceil}^d \mathbb{P}(\operatorname{Bin}(d, p)=j)\right)$$ as $\sum_{i=1}^\infty(a_ib_i)$ for which we know that $\sum_{i=1}^\infty a_i <1$ (fact 1) and for every $i\in \mathbb{N}$, $b_i<\frac{1}{2}$ (fact 2).  Hence $\sum_{i=1}^\infty(a_ib_i) <\frac{1}{2}$. \\
   Therefore, $$
   \sum_{d=1}^{\infty}\left(\mathbb{P}(\operatorname{Poisson}(c)=d)\cdot\sum_{j=\lceil\frac{d+1}{2}\rceil}^d \mathbb{P}(\operatorname{Bin}(d, p)=j)\right)<\frac{1}{2}.
   $$
\end{proof}
We will prove Equation \eqref{eq:red_maj} by first showing that, if $c<x_0$, $CDF_{\operatorname{Poisson}(x_0)}(\mathbf{P})<\frac{1}{2}$.

\begin{proposition}\label{prop:CDF(P)<0.5}
    \begin{equation*}
        CDF_{\operatorname{Poisson}(x_0)}(\mathbf{P})<\frac{1}{2}
    \end{equation*}
    (where  $\mathbf{P}\sim \operatorname{Poisson}(c)$),    whenever $c<x_0$.
\end{proposition}

To prove Proposition \ref{prop:CDF(P)<0.5}, we make use of the following lemma:
\begin{lemma}\label{lem:P(X<=Y)}
    For two independent distributions $X$ and $Y$ that both take integer values, $$\sum_{k=0}^\infty\left(\mathbb{P}(X\leq k)\cdot\mathbb{P}(Y=k)\right) = \mathbb{P}(X\leq Y)$$
\end{lemma}
\begin{proof}[Proof of Lemma \ref{lem:P(X<=Y)}]
Starting from the right-hand side, we can compute the probability that $\mathbb{P}(X\leq Y)$ by summing over all possible values of $Y$:
$        \mathbb{P}(X\leq Y)=\sum_{k=0}^\infty\left(\mathbb{P}(X\leq Y|Y=k)\cdot\mathbb{P}(y=k)\right) $
    Since for each integer $k$,  $\mathbb{P}(X\leq Y|Y=k)=\mathbb{P}(X\leq k)$, we have the desired equality $\sum_{k=0}^\infty\left(\mathbb{P}(X\leq k)\cdot\mathbb{P}(y=k)\right) = \mathbb{P}(X\leq Y).$
\end{proof}
\begin{proof}[Proof of Proposition \ref{prop:CDF(P)<0.5}]
We can write $CDF_{\operatorname{Poisson}(x_0)}(\mathbf{P})$ as follows:
\begin{equation*}
    \begin{aligned}
            CDF&_{\operatorname{Poisson}(x_0)}(\mathbf{P}) \\
            &= \sum_{k=0}^\infty\left( CDF_{\operatorname{Poisson}(x_0)}(k)\cdot \mathbb{P}(\operatorname{Poisson}(c)=k) \right)\\
            &=\sum_{k=0}^\infty\left(\mathbb{P}(\operatorname{Poisson}(x_0)\leq k)\cdot\mathbb{P}(\operatorname{Poisson}(c)=k)\right)
    \end{aligned}
\end{equation*}
Let $X$ and $Y$ be independent random variables with $X\sim \operatorname{Poisson}(x_0)$ and $Y\sim \operatorname{Poisson}(c)$. We want to show that $\sum_{k=0}^\infty\left(\mathbb{P}(X\leq k)\cdot\mathbb{P}(Y=k)\right)<\frac{1}{2}$ whenever $c<x_0$. By Lemma \ref{lem:P(X<=Y)}, this is equivalent to showing that $\mathbb{P}(X\leq Y)<\frac{1}{2}$.
We decompose $X$ into two independent distributions $Z$ and $W$, so $X=Z+W$ such that $Z\sim \operatorname{Poisson}(c)$ and $W\sim \operatorname{Poisson}(x_0-(c))$.
We first compare $Z$ and $Y$. Since they are independent and identically distributed, we have that $\mathbb{P}(Z\leq Y) = \frac{1}{2}+\frac{1}{2}\mathbb{P}(Z=Y)$, and $\mathbb{P}(Z<Y) = \frac{1}{2}-\frac{1}{2}\mathbb{P}(Z=Y)$.

Now, we consider the probability of $\{X\leq Y\}$ based on the different possible values of $W$:
\begin{equation}\label{eq:X<=Y}
\begin{aligned}    
    \mathbb{P}(X\leq Y) &= \mathbb{P}(Z+W\leq Y) \\
    &= \sum_{w=0}^\infty\left( \mathbb{P}(W=w)\cdot\mathbb{P}(Z+w\leq Y) \right).    
\end{aligned}
\end{equation}
\begin{itemize}
    \item For $w=0$, $\mathbb{P}(Z+w\leq Y)=\mathbb{P}(Z\leq Y)=\frac{1}{2}+\frac{1}{2}\mathbb{P}(Z=Y)$.
    \item For $w\geq 1$, $Z+w\leq Y$ implies $Z < Y$, so $\mathbb{P}(Z+w\leq Y)\leq \mathbb{P}(Z<Y) = \frac{1}{2}-\frac{1}{2}\mathbb{P}(Z=Y)$
\end{itemize}
Since $W\sim \operatorname{Poisson}(x_0-(c))$, we know the probability that $W$ is 0: $\mathbb{P}(W=0)= e^{-(x_0-(c))}$. Now we can fill in the above probabilities in Equation \eqref{eq:X<=Y}. For readability we write $\mathbb{P}(Z=Y) = p_{zy}$ and $\mathbb{P}(W=0)=p_0$
\begin{equation*}
    \begin{aligned}
        \mathbb{P}(X&\leq Y) = \mathbb{P}(Z+W\leq Y) \\
        &= \sum_{w=0}^\infty\left( \mathbb{P}(W=w)\cdot\mathbb{P}(Z+w\leq Y) \right)\\
        &\leq p_0\cdot \mathbb{P}(Z\leq Y) + \sum_{w=1}^\infty\left(\mathbb{P}(W=w)\mathbb{P}(Z<Y)\right)\\
        & = p_0\cdot(\frac{1}{2}+\frac{1}{2}p_{zy}) + \sum_{w=1}^\infty\left(\mathbb{P}(W=w)\cdot(\frac{1}{2}-\frac{1}{2}p_{zy})\right)\\
        & = p_0\cdot(\frac{1}{2}+\frac{1}{2}p_{zy}) +(1-p_0)\cdot(\frac{1}{2}-\frac{1}{2}p_{zy}))\\
        & =\frac{1}{2}p_0+\frac{1}{2}p_0p_{zy} + \frac{1}{2} - \frac{1}{2}p_{zy} -\frac{1}{2}p_0 + \frac{1}{2}p_0p_{zy} \\
        & = \frac{1}{2}+\frac{1}{2}p_{zy}(2 p_0-1)
    \end{aligned}
\end{equation*}
We already know that  $p_0=e^{-(x_0-c)}$. Since $c<x_0$ and both $c$ and $x_0$ are integer, we have that $c\leq x_0-1$, so $x_0-c\geq 1$. Because $e^{-1}<\frac{1}{2}$ already, and $e^{-(x_0-c)}$ must be even smaller,  we have that $e^{-(x_0-c)}<\frac{1}{2}$. Hence, $2 p_0-1<0$. Since $p_{zy}=\mathbb{P}(Z=Y)>0$, we have that $\frac{1}{2}p_{zy}(2 p_0-1)<0$. Hence, $\mathbb{P}(X\leq Y)\leq \frac{1}{2}+\frac{1}{2}p_{zy}(2 p_0-1)<\frac{1}{2}$.
\end{proof}

Since Proposition \ref{prop:CDF(P)<0.5} tells us that, whenever\\ $c<x_0$, $CDF_{\operatorname{Poisson}(x_0)}(\mathbf{P})<\frac{1}{2}$, we have that the probability that $o$ sees a red majority in the Poisson tree colored according to  $\gamma_{d,\pi}$ (Equation \ref{eq:o_sees_red_majority_ydegree}) is equal to the probability that $o$ sees a red majority in the Poisson tree where every node is colored red with probability $p<\frac{1}{2}$ (Equation \ref{eq:v_sees_red_majority}) which is less than $\frac{1}{2}$ according to Lemma \ref{lem:v_sees_red_majority<0.5}.

Therefore, we can conclude that $$\mathbb{P}(o \text{ sees red maj. in }T^{\gamma_{d,\pi}}_{\operatorname{Poisson}(c)})<\frac{1}{2}$$ for any $c$ and $x_0 >c$, which implies  (Lemma \ref{lem:conv_Mm_ydegree}) that the probability of a \texttt{Mm} illusion converges to 0.

\end{proof}

\subsection{Illustration of CDFs of Poisson distributions}\label{app:PoissonCDFs}

\begin{figure}[h]
    \centering
    \includegraphics[width=\linewidth]{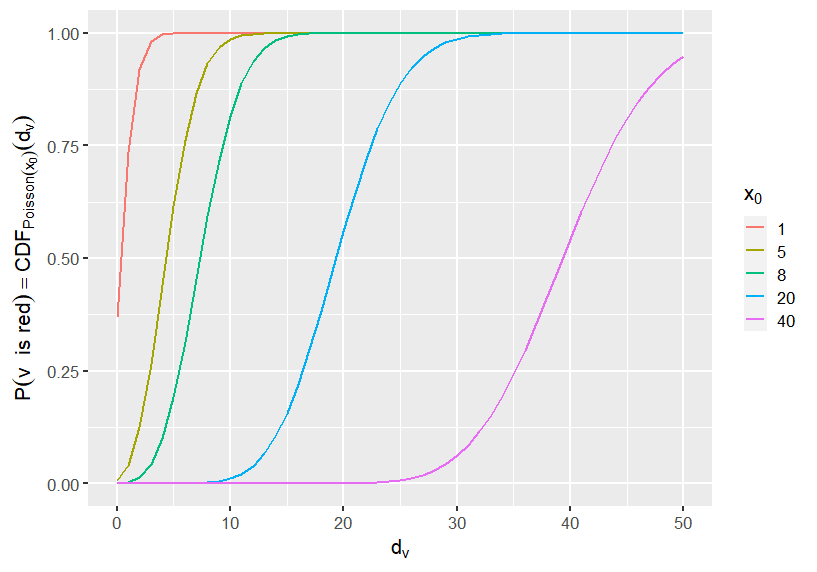}
    \caption{Examples of cumulative distribution functions of $\operatorname{Poisson}(x_0)$ distributions with different values of $x_0$.}
    \label{fig:Poisson_examples}
\end{figure}

\end{document}